\documentclass[final, 4p, 12pt]{elsarticle}
\usepackage{amssymb}
  \usepackage{amsmath}
  \usepackage{amsthm}
   \theoremstyle{plain}
  \usepackage{bm}

  \usepackage{verbatim}
  \usepackage{textcomp} % \texttimes, \textdegree, \textohm, \textmu
\usepackage{array}
\def\foo#1\endfoo{}
\newcolumntype{L}{@{}>{\foo}l<{\endfoo}}
\usepackage{booktabs} % improved rules (lines) for tables
\usepackage{dcolumn}  % align at decimal in tables
\usepackage{multirow} % table elements spanning multiple rows

\usepackage{lettrine}
\usepackage{algorithm}
\usepackage{algorithmicx}
\usepackage{algpseudocode}
\makeatletter \renewcommand*{\ALG@name}{Method}
\usepackage{mathrsfs}  

\usepackage{simplemargins}
\setallmargins{.8in}
\newtheorem{thm}{Theorem}
\makeatletter
\renewcommand*\env@matrix[1][c]{\hskip -\arraycolsep
  \let\@ifnextchar\new@ifnextchar
  \array{*\c@MaxMatrixCols #1}}
\makeatother

\usepackage{url} 

\newtheorem{lem}{Lemma}

\theoremstyle{definition}
\newtheorem{defn}{Definition}
\newtheorem*{pf}{Proof}

\DeclareFontFamily{U}{matha}{\hyphenchar\font45}
\DeclareFontShape{U}{matha}{m}{n}{
      <5> <6> <7> <8> <9> <10> gen * matha
      <10.95> matha10 <12> <14.4> <17.28> <20.74> <24.88> matha12
      }{}
\DeclareSymbolFont{matha}{U}{matha}{m}{n}
\DeclareFontFamily{U}{mathx}{\hyphenchar\font45}
\DeclareFontShape{U}{mathx}{m}{n}{
      <5> <6> <7> <8> <9> <10>
      <10.95> <12> <14.4> <17.28> <20.74> <24.88>
      mathx10
      }{}
\DeclareSymbolFont{mathx}{U}{mathx}{m}{n}

\DeclareMathSymbol{\obot}         {2}{matha}{"6B}
\DeclareMathSymbol{\bigobot}       {1}{mathx}{"CB}
\newcommand{\oant}{\mbox{OA}(N,k,2,t)}
\newcommand{\oantf}{\mbox{OA}(N,k,2,4)}

\newcommand{\oan}{{\rm OA}(N,k,s,t)}

\newcommand{\1}{{\bf 1}}

\newcommand{\bb}{{\bf{b}}}
\newcommand{\cc}{{\bf c}}

\newcommand{\uu}{{\bf u}}

\newcommand{\x}{{\bf x}}

\newcommand{\vv}{{\bf v}}

\newcommand{\y}{{\bf y}}

\newcommand{\Y}{{\bf Y}}

\newcommand{\X}{{\bf X}}

\newcommand{\B}{{\bf B}}

\newcommand{\ddd}{{\bf d}}
\newcommand{\D}{{\bf D}}

\newcommand{\A}{{\bf{A}}}

\newcommand{\etal}{\emph{et al.}}

\newcommand{\GLPP}{G^{{\rm LP(\ref{eqn:ILPOAfin})}}}
\usepackage{hyperref}
\newcommand{\GLPfin}{G^{\rm LP(\ref{eqn:ILPOAfin})}}
\begin{document}
%-----------------------------------------------------------------------
\begin{frontmatter}
\title{Parallelizing the branch-and-bound with isomorphism pruning algorithm for classifying orthogonal arrays}
\author[AFIT]{Dursun A.~Bulutoglu}%\corref{cor1}}
\ead{dursun.bulutoglu@gmail.com}
%\author[]{}
%\ead{}
\address[AFIT]{Department of Mathematics and Statistics, Air Force Institute of Technology,\\Wright-Patterson Air Force Base, Ohio 45433, USA}
%\address[]{ }
\journal{Operations Research Letters}
%-----------------------------------------------------------------------
\begin{abstract}
We provide a method for parallelizing the branch-and-bound with isomorphism pruning algorithm developed by Margot \, [Symmetric ILP: Coloring and small integers, Discrete Optimization (4) (2007), 40-62]. We apply our method to classify orthogonal arrays. For classifying all non-OD-equivalent OA$(128,9,2,4)$ and OA$(144,9,2,4)$ our method results in linear speedups. Finally, our method enables classifying all non-OD-equivalent  OA$(192,k,2,4)$ for $k=9,10,11$ for the first time.
\end{abstract}
\begin{keyword}
Breadth first search; Integer linear program;  LP relaxation symmetry group;   {\tt Perl} programming language; Permutation group action
\end{keyword}
\end{frontmatter}
%--------------------------------------------------------------------------
\section{Introduction}
%------------------------------------------------------------
A branch-and-bound (B\&B) algorithm is used to find one or all optimum solutions to an integer linear program (ILP) with $n$ variables of the form
\begin{equation}\label{eqn:ILP}
\begin{aligned}
& \quad \quad \quad \quad \min  \quad \cc^{\top}\x  \\
& \quad  \A \x=\bb,\quad \B\x\leq  \ddd, \quad \x \in 
\mathbb{Z}^n, 
& \quad \quad \quad \quad \quad \ \  
\end{aligned}
\end{equation}
where $\A$ and $\B$ are $m_1 \times n$ and $m_2 \times n$
constraint matrices, and $\cc \in \mathbb{R}^n$, $\bb\in \mathbb{R}^{m_1}$,
 $\ddd \in \mathbb{R}^{m_2}$. A B\&B algorithm recursively divides ILP~(\ref{eqn:ILP}) into subproblems by branching on the values of its integer variables. This results in the {\em B\&B search tree}, where each node is a subproblem. In a B\&B algorithm, the branching rule that selects the minimum index non-fixed variable for branching is called {\em minimum index branching}. Throughout the paper, unless otherwise stated,  minimum index branching is the only branching rule
 that is used or discussed.
 
 At each node, a B\&B algorithm solves the linear programming (LP) relaxation obtained by dropping the integrality of the variables. Then, a node is pruned
 if either the LP relaxation is infeasible or the optimum objective value of the LP relaxation is larger than that of a previously known  feasible solution. 
 Such an algorithm can either be stopped as soon as an optimum solution based on a known global lower bound is found or run until all optimum solutions are obtained. Since an LP relaxation needs to be solved at every node of a  B\&B search tree,
 the number of nodes in the tree, i.e., the size of the tree, dictates the CPU time required to find all optimum solutions, where smaller-sized trees require shorter CPU time.
In fact, the required CPU time depends roughly linearly on the
number of nodes~\cite{cornuejols2004}. On the other hand, 
the size of a B\&B search tree depends on the order of the variables or, more generally, on the branching rule of the B\&B algorithm~\cite{Linderoth1999}. Shrinking the size of a B\&B search tree by developing new branching strategies is an active area of research~\cite{Aardal2024}.
The size of the B\&B search tree also depends on the formulation of the ILP. In~\cite{Pataki2009}, there are many examples of $50$ variable knapsack ILPs for which a
reformulation based on a change of variables resulting from the Lenstra–Lenstra–Lov\'asz
(LLL) lattice basis reduction algorithm decreases the B\&B search tree size by a factor of
at least 10 million.

The set of permutations of the variables of an ILP whose elements map each feasible solution to another feasible solution with the same objective function value is called the {\em symmetry group} of the ILP.  For ILPs with large symmetry groups, the B\&B algorithm processes a large number of nodes with identical feasibility statuses and LP relaxation objective function values.  To decrease the size of the B\&B search  tree for such problems, Margot~\cite{Margot2007} developed the B\&B with isomorphism pruning algorithm. The B\&B with isomorphism pruning algorithm uses a subgroup of the symmetry group of the ILP for additional pruning, resulting in a smaller B\&B search tree. Using a larger subgroup is more desirable, as it allows for additional pruning, resulting in a smaller B\&B search tree and, consequently, decreased solution times. For a subgroup $G$ of the symmetry group of ILP~(\ref{eqn:ILP}), two solutions $\x_1$ and $\x_2$ of ILP~(\ref{eqn:ILP}) are called $G$-isomorphic if there is a  $g \in G$ such that $g(\x_1)=\x_2$. When the B\&B with isomorphism pruning algorithm is used to find all optimum solutions of  ILP~(\ref{eqn:ILP}), it provides a set of all
non-$G$-isomorphic optimum solutions.

The largest subgroup of the symmetry group of an ILP for which there is a known practical method to obtain is the LP relaxation symmetry group. The {\em LP relaxation symmetry group}, $G^{\rm LP}$, of an ILP is the set of all permutations of its variables that send each feasible point of its LP relaxation to another feasible point with the same objective function value. Geyer \etal~\cite{Geyer2018} developed a practical method for finding $G^{\rm LP}$  by determining the automorphism group of a particular graph obtained from the ILP. In this paper, we develop a parallelization method for the B\&B algorithm with isomorphism pruning and apply it to classifying orthogonal arrays using the $\GLPP$ of an orthogonal array defining ILP.

\begin{defn}
 An $N$ row, $k$ column, $s$-symbol, and strength $t$ orthogonal array OA$(N,k,s,t)$  is an $N\times k$ array of $s$ symbols such that 
 every subset of $t$ columns contains every $s^t$ symbol combination as a row exactly $N/s^t$ times.
\end{defn}
For an $N\times k$ array $\D$ whose entries belong to a set of $s$ symbols $\{l_0,\ldots,l_{s-1}\}$, and an integer $i$ such that $0\leq i\leq s^k-1$, let $x_i$ be the number of times the symbol combination $(l_{i_1},\ldots,l_{i_k})$ appears as a row in the array, where $(i_1,\ldots,i_k)_s$ is the base $s$ representation of $i$. Then the vector
$\x \in \{0,\ldots, N\}^{s^k}$ is called the {\em frequency vector} of $\D$. 
Let $\x$ be the frequency vector of an 
$\oan$ and $\lambda=N/s^t$.
Bulutoglu and Margot~\cite{Bulutoglu2008} provided the ILP formulation 
\begin{equation}\label{eqn:ILPOA}
\begin{aligned}
& \quad \quad \quad \quad \min  \quad \1^{\top} \x  \\
& \quad  \A(s,k,t) \x=\lambda \1, \quad \x \in \{0,\ldots, p_{\max}\}^{s^k}, 
\end{aligned}
\end{equation}
where $\A(s,k,t)$ and $p_{\max}$ are obtained as in~\cite{Bulutoglu2008}, and $\1$ is a vector of all $1$s of appropriate size. 
Then each solution of ILP~(\ref{eqn:ILPOA})
is the frequency vector of an $\oan$. Since $\1^{\top}\x=N$, every solution to ILP~(\ref{eqn:ILPOA})
is optimum. Let the rows of $\A'$ be a subset of the rows of $\A$ such that they form a basis for the row space of $\A$. Rows of $\A'$ are obtained from $\A$ by using Gaussian elimination. 
Then each solution to 
\begin{equation}\label{eqn:ILPOAfin}
\begin{aligned}
& \quad \quad \quad \quad \min  \quad 0  \\
& \quad  \A'(s,k,t) \x=\lambda\1, \quad \x \in \{0,\ldots,p_{\max}\}^{s^k}, 
\end{aligned}
\end{equation}
is an $\oan$. Even though ILP~(\ref{eqn:ILPOA}) and ILP~(\ref{eqn:ILPOAfin}) have the same solution sets, 
obtaining all solutions of ILP~(\ref{eqn:ILPOAfin}) by B\&B is faster than that of  
ILP~(\ref{eqn:ILPOA}). This is because ILP~(\ref{eqn:ILPOAfin}) is sparser, and the LP relaxations solved in a B\&B algorithm for sparser ILPs take less time to solve, resulting in a faster overall algorithm. For several $N,k,2,t$ combinations, we were able to find sparser formulations of
ILP~(\ref{eqn:ILPOAfin}) by applying the LLL reduction based reformulation method in~\cite{Pataki2009}. However, none of these reformulations allowed the B\&B algorithm to find all solutions faster. We also tried changing the order of the variables by using different strategies. However, none of the order changes or branching strategies improved the speed of finding all solutions 
compared to the minimum index B\&B for ILP~(\ref{eqn:ILPOAfin}), where variables are lexicographically ordered.

Each of the operations of permuting rows, columns, and symbols within any column of an $N$ row, $k$ column, $s$-symbol array is called an {\em isomorphism operation}~\cite{Geyer2018}. Isomorphism operations generate a group called the {\em isomorphism group} denoted by $G^{\rm iso}$(k,s)~\cite{Geyer2018}. Two $k$ column, $s$-symbol, $N$ row arrays are called {\em isomorphic} if one can be obtained from the other by applying a sequence of isomorphism operations.

%For $2$-symbol $N \times k$ arrays, we use the symbols $\{0,1\}$. This allows us to define additional operations for $2$-symbol arrays. 
For two vectors $\uu_1,\uu_2 \in \{0,1\}^N$,
let $$\uu_1 \oplus \uu_2= \uu_1+\uu_2 \mod 2.$$
For two vectors $\vv_1,\vv_2 \in \{-1,1\}^N$,
let $\vv_1 \odot \vv_2$ be the Hadamard product (component-wise product) of $\vv_1$ and $\vv_2$.
Let $\X$ be an $N \times k$ array with symbols from \{-1,1\}. Then,  
for $m \in \{1,\ldots,k\}$, Geyer \etal~\cite{Geyer2018} defined the column operation $R_m$ on $\X$ to be
{\small
\begin{equation*}\label{eqn:MZZX}
\X=\begin{bmatrix}
\x_1& \cdots &  \x_m&\cdots&\y_k\\
 \end{bmatrix} \quad  \stackrel{R_m\,\,\,}{\longrightarrow}\quad
 \begin{bmatrix}
\x_1\odot\x_m& \cdots&\x_{m-1}\odot\x_m& \x_m&\x_{m+1}\odot\x_m&\cdots & \x_k\odot \x_m\\
 \end{bmatrix}.
\end{equation*}}
For $\Y$  an $N \times k$ array with symbols from \{0,1\}, and  
 $m \in \{1,\ldots,k\}$, define the column operation $R'_m$ on $\Y$ to be
{\small
\begin{equation*}\label{eqn:MZZY}
\Y=\begin{bmatrix}
\y_1& \cdots &  \y_m&\cdots&\y_k\\
 \end{bmatrix} \quad  \stackrel{R'_m\,\,\,}{\longrightarrow}\quad
 \begin{bmatrix}
\y_1\oplus\y_m& \cdots&\y_{m-1}\oplus\y_m& \y_m&\y_{m+1}\oplus\y_m&\cdots & \y_k\oplus \y_m\\
 \end{bmatrix}.
\end{equation*}}
Let $\X=(-1)^{\Y}$ or $\log_{-1}(\X)=\Y$, where $x_{ij}=(-1)^{y_{ij}}$.
Then, $\X$ is obtained from $\Y$ by replacing ``$1$" with ``$-1$" and ``$0$" with ``$1$", and 
\begin{equation}\label{eqn:connect}
R_m(\X)=(-1)^{R'_m(\Y)}.
\end{equation}
For a set $S$, let $\langle S \rangle$ be the group generated by the elements of $S$.
Let $$G(k)^{\rm OD}_1=\langle R_1,\ldots,R_k, G^{\rm iso}(k,2)\rangle,  \text{  }  G(k)^{\rm OD}_2=\langle R'_1,\ldots,R'_k, G^{\rm iso}(k,2)\rangle.$$
Then, by equation~(\ref{eqn:connect}),
\begin{equation*}\label{eqn:connect2}
g(\X)=(-1)^{g'(\Y)}.
\end{equation*}
for all $g \in G(k)^{\rm OD}_1 $ and $g' \in G(k)^{\rm OD}_2$. Hence, $G(k)^{\rm OD}_1 \cong G(k)^{\rm OD}_2$. For the remainder of the paper, we use  $G(k)^{\rm OD}_2$ and refer to it as $G(k)^{\rm OD}$.
%\begin{equation}\label{eqn:Gkod}
%=\langle \rangle.
%\end{equation} 

Each of the operations in $G(k)^{\rm OD}$ is called an {\em OD-equivalence operation}~\cite{Geyer2018}. 
Two $k$ column, $2$-symbol, $N$ row arrays are called {\em OD-equivalent} if one can be obtained from the other by applying a sequence of OD-equivalence operations~\cite{Geyer2018}.
It is well known that if an $N$ row, $k$ column, and $s$-symbol array $\D$ is isomorphic to an $\oan$, then $\D$ is also an $\oan$. Geyer \etal~\cite{Geyer2018} showed that if $\D$ is OD-equivalent to an $\oant$ for even $t$, then $\D$ is also an $\oant$. However, this statement does not hold for odd $t$~\cite{Geyer2018}.

For $g \in G(k)^{\rm OD}$ and $i=(i_1,\ldots,i_k)_s$, let $g(i_1,\ldots,i_k)=(i'_1,\ldots,i'_k)$. Then define $g(x_i)=x_j$, where $j=(i'_1,\ldots,i'_k)_s$.
Hence, the group  $G(k)^{\rm OD}$  permutes the variables in  $\{x_0,\ldots x_{s^k-1}\}$
by permuting $(i_1,\ldots,i_k)_s$ for $i_r\in \{0,\ldots,s-1\}$ and $r \in \{1,\ldots,k\}$.
Let $\GLPfin$ be the LP relaxation symmetry group of ILP~(\ref{eqn:ILPOAfin}). Geyer \etal~\cite{Geyer2018}
showed that 

\begin{equation}\label{eqn:containment}
\GLPfin \geq
\begin{cases}
 G(k)^{\rm OD} & \text{if $t$ is even and $s=2$,}  \\
G^{\rm iso}(k,s) & \text{otherwise.} 
\end{cases}
\end{equation}
%$$\GLPfin \geq G^{\rm iso}(k,s)$$ in general, and $$\GLPfin \geq  $$ if $t$ is even and $s=2$. 
Computational evidence suggests that for $1 \leq t\leq k-1$, the group containments 
in~(\ref{eqn:containment}) are, in fact, equalities.  So, we conjecture this to be the case and assume it to be true throughout the paper.

%For $R'_{m} \in  G(k)^{\rm OD}$ and $i=(i_1,i_2,\ldots,i_k)_s$, let $R'_m(i_1,i_2,\ldots,i_k)=(i'_1,i'_2,\ldots,i'_k)$. Then define $R'_m(x_i)=x_j$ , where $j=(i'_1,i'_2,\ldots,i'_k)_s$.
%Hence, the group  $G^{\rm iso}(k,s)$  permutes the variables in  $\{x_0,\ldots x_{s^k-1}\}$
%by permuting $(i_1,i_2,\ldots,i_k)_s$ for $i_r\in \{0,\ldots,s-1\}$ for $r \in \{1,\ldots,k\}$.

In this paper, we develop a method for parallelizing the B\&B with isomorphism pruning algorithm and apply our method to classify all non-OD-equivalent $\oantf$.
%In Section~\ref{sec:BandBiso}, 
In Section~\ref{sec:Parallel}, we describe the B\&B with isomorphism pruning algorithm and develop our parallelization method.
In Section~\ref{sec:OAs}, we use the parallelized  B\&B with isomorphism pruning algorithm with the group $\GLPfin$ to classify all non-OD-equivalent $\oantf$. In particular,  for classifying all non-OD-equivalent OA$(128,9,2,4)$ and OA$(144,9,2,4)$ our method results in super-linear speedups. Moreover, it allows classifying all non-OD-equivalent  OA$(192,k,2,4)$ for $k=9,10,11$ for the first time.
%-----------------------------------------------------------------
 \section{Parallelizing the B\&B with isomorphism pruning algorithm}\label{sec:Parallel}
%----------------------------------------------------------------------
In this section, we describe the B\&B with isomorphism pruning algorithm.  If ILP~(\ref{eqn:ILP}) has an unbounded variable $x_j$, then we can write 
$x_j=x_{j1}-x_{j2}$ for variables $x_{j1}\geq 0$ and $x_{j2}\geq 0$. If $x_j \geq a$ for some $a\leq 0 $, then $x'_j=x_j-a\geq 0$.  
Hence, WLOG we assume that the variables of ILP~(\ref{eqn:ILP}) are non-negative. Let $u_i$ be the upper bound of $x_i$ obtained by solving the LP obtained 
by replacing the objective function with $-x_i$ and dropping the integrality constraints on the variables.

We call $\x$ a {\em partial solution} of ILP~(\ref{eqn:ILP}) if  the elements of a subset of the entries of $\x$ are fixed to elements of $\mathbb{R}$.  Variables that are not fixed are kept as indeterminates. To keep the notation simpler, we also call every element of $\mathbb{R}^n$  a partial solution.
For two partial solutions $\x_1$ and $\x_2$ of ILP~(\ref{eqn:ILP}), let     $\x_1'$ and $\x_2'$ be obtained from $\x_1$ and $\x_2$ by replacing unfixed variables with $-1$.
Then, we say that $\x_1$ is {\em lexicographically smaller} than $\x_2$ $(\x_1 \prec \x_2)$ if the first non-zero entry in $\x_1'-\x_2'$ is positive.

A group $G$ with identity $e$ and a function $\cdot$
such that $\cdot : G \times S \rightarrow S$  {\em acts on} a set $S$ if 
\begin{enumerate}
    \item $e\cdot s=s$ $\forall s \in S$,
    \item $(g_1g_2)\cdot s=g_1 \cdot (g_2 \cdot s)$ $\forall g_1,g_2 \in G$ and $s\in S.$
\end{enumerate}
For $x \in S$, the set $$G\,x=\{y \in S\mid y=g\cdot x \text{ for some } g \in G \}$$
is called {\em the orbit
of $x$ under the action of $G$} or a {\em $G$-orbit}. 

Let $G$ be a  permutation group that permutes the variables of   ILP~(\ref{eqn:ILP}) by permuting the indices of its variables.
Then, 
$$g(\x)_j=\x_{g^{-1}(j)}.$$
% Let $S$ be the set of all permutations of the variable vector $\x$ of ILP~(\ref{eqn:ILP}) and
Let $\mathbb{R}^n_{\x}$ be the set of all partial solutions of   
ILP~(\ref{eqn:ILP}).
 Define
  $\cdot : G \times \mathbb{R}^n_{\x} \rightarrow \mathbb{R}^n_{\x}$ to be
 $g\cdot \x=g(\x)$ for all $g \in G$ and $\x \in \mathbb{R}^n_\x$. Then, for the identity permutation $id$, 
 $$id\cdot \x=id(\x)=\x$$ and
$$(g_1g_2)\cdot \x=(g_1g_2)(\x)=g_1(g_2(\x))=g_1 \cdot (g_2 \cdot \x).$$
This is because $$(g_1g_2)(\x)_j=\x_{(g_1g_2)^{-1}(j)}=\x_{(g_2^{-1}g_1^{-1})(j)}=(g_1(g_2(\x))_j.$$
Hence,  $G$ acts on partial solutions. 
This action splits the set of partial solutions into $G$-orbits. Let the lexicographically minimum partial solution in each $G$-orbit be the representative of its orbit. Such a partial solution is called 
{\em lexicographically-$G$-minimum}.

In a B\&B algorithm, by {\em evaluating a node}, we mean solving the LP relaxation at the node.
 At a B\&B search tree node, the number of variables fixed by branching decisions is called the {\em depth of the node}. 
A branching algorithm that always creates child nodes before evaluating sibling nodes is called {\em depth first search}. A branching algorithm that creates child nodes only after evaluating all the sibling nodes is called {\em breadth first search}. In a breadth first search, the search reaches the $r$th {\em level} when the depth of every unpruned leaf is $r$. 

Let $G$ be a subgroup of the symmetry group of ILP~(\ref{eqn:ILP}).
 The {\em B\&B  with isomorphism pruning algorithm} of Margot~\cite{Margot2007} implements depth first search minimum index branching and removes a partial solution
 (node) of the B\&B search tree if it is not lexicographically minimum within its orbit under the action of $G$. 
 %Such a B\&B algorithm is called the {\em B\&B algorithm with isomorphism pruning}.
  Then, all feasible leaves of the resulting B\&B search tree are a set of all lexicographically-$G$-minimum, $G$-non-isomorphic optimum solutions.  
%%%%%%%%%%%%%%%%%%%%%%%%%%%%%%%%%%%%%%%%%%%%%%%%%%%%%%%%%%%%%%%%%%%%%%
\begin{algorithm}[t!]
\centering
\caption{\label{alg:determine} Determining jobs to run in parallel }
\begin{algorithmic}[1]
\State {\bf Input} An ILP $\mathcal{P}$, a subgroup $G$ of its symmetry group, time $T$ in seconds, and $m$.
\State  {\bf Set} $iter=1$;
\State {\bf Label} $\mathcal{P}$ unfinished;
\While {$iter \leq m $} 
 \State {\bf Branch} on the values $\{0,\ldots,u_{iter}\}$ of the minimum index non-fixed variable $x_{iter}$
 of each subproblem that is unfinished;~\label{step:branch}
\State {\bf Run} the B\&B with isomorphism pruning algorithm with group $G$ on each resulting subproblem from Step~\ref{step:branch} for $T$ seconds;~\label{step:run}
\State {\bf Label} a subproblem run in Step~\ref{step:run}  unfinished if all of its optimum solutions could not be found in  $T$ seconds, otherwise label it finished;\\
$iter=iter+1$;
\EndWhile
\State {\bf Output} A set of subproblems ILP$_1,\ldots,$ILP$_r$ to be solved in parallel to solve the ILP $\mathcal{P}$ and the set $S$ of
%a file called sol that contains 
all non-$G$-isomorphic optimum solutions found from subproblems labeled finished.
\end{algorithmic}
\end{algorithm}
%%%%%%%%%%%%%%%%%%%%%%%%%%%%%%%%%%%%%%%%%%%%%%%%%%%%%%%%%%%%%%%%%%%%
%-----------------------------------------------------------------

 In Methods~\ref{alg:determine}-\ref{alg:implement}, we assume that $u_i <\infty$ for all $i\in \{1,\ldots,n\}$.
Method~\ref{alg:determine} determines bottleneck subproblems to run the  B\&B with isomorphism pruning algorithm in parallel by implementing a breadth first search up to level $m$.
Method~\ref{alg:runjobs} 
runs the  B\&B with isomorphism pruning algorithm in parallel on the subproblems determined by 
Method~\ref{alg:determine}. Method~\ref{alg:implement}  parallelizes by implementing  Method~\ref{alg:determine} followed by  Method~\ref{alg:runjobs}. Hence, Method~\ref{alg:implement} executes a breadth first search up to level $m$ and a depth first search beyond level $m$.

%--------------------------------------------------------------------------
\begin{algorithm}[t!]
\centering
\caption{\label{alg:runjobs} Run jobs in parallel}
\begin{algorithmic}[1]
\State {\bf Input} Subproblems ILP$_1,\ldots,$ILP$_r$ and the set $S$ of all non-$G$-isomorphic optimum solutions to the ILP $\mathcal{P}$ found by Method~\ref{alg:determine}.
\For {$i:=1$ {\bf to} $r$ {\bf step} $1$}
\State {\bf Run} the B\&B with isomorphism pruning algorithm with group $G$ on ILP$_i$ in parallel;
\EndFor
\State {\bf Combine} all the obtained solutions and the set of all non-$G$-isomorphic optimum solutions to the ILP $\mathcal{P}$ found by Method~\ref{alg:determine} into a single file;  
\State {\bf Output} The set of all lexicographically-$G$-minimum, non-$G$-isomorphic optimum solutions to ILP $\mathcal{P}$.
\end{algorithmic}
\end{algorithm}
%--------------------------------------------------------------------------
\begin{algorithm}[t!]
\centering
\caption{\label{alg:implement} Implement parallelization}
\begin{algorithmic}[1]
\State {\bf Input} An ILP $\mathcal{P}$, a subgroup $G$ of its symmetry group, time $T$ in seconds, and $m$.%~\label{alg:implement}
\State {\bf Use} Method~\ref{alg:determine} with inputs  $\mathcal{P}$, $G$, $T$, and $m$ to obtain subproblems ILP$_1,\ldots,$ILP$_r$, and the set $S$ of all non-$G$-isomorphic optimum solutions to the ILP $\mathcal{P}$ found by Method~\ref{alg:determine};~\label{step:determine}
\State {\bf Use}   Method~\ref{alg:runjobs} with inputs ILP$_1,\ldots,$ILP$_r$, and  $S$;~\label{step:finish}
\State {\bf Output}  The set of all lexicographically-$G$-minimum, non-$G$-isomorphic optimum solutions to ILP $\mathcal{P}$.
\end{algorithmic}
\end{algorithm}
The following lemma is needed to justify  that the output of  Method~\ref{alg:implement} is correct.
\begin{lem}\label{lem:parent}
  The parent node of any node in the B\&B with isomorphism pruning algorithm search tree  must be lexicographically minimum under the action of $G$. 
\end{lem}
\begin{pf}
   The  B\&B with isomorphism pruning algorithm removes nodes that are not lexicographically minimum under the action of $G$. Hence, no node whose parent node is not lexicographically minimum under the action of $G$ can exist in the B\&B with isomorphism pruning algorithm search tree.
   \qed
\end{pf}
The following theorem justifies that the output of  Method~\ref{alg:implement} is correct.
\begin{thm}\label{thm:main}
The output of Method~\ref{alg:implement} is  a set of all non-$G$-isomorphic optimum solutions to ILP $\mathcal{P}$.   
\end{thm}
\begin{pf}
It suffices to prove that it is viable to use the group $G$ with the B\&B with isomorphism pruning algorithm for each subproblem obtained by branching. If a node (subproblem) obtained by branching 
is not lexicographically minimum under the action of $G$, and the B\&B with isomorphism 
pruning algorithm is run with the group $G$ on that subproblem, then the algorithm declares the subproblem to be infeasible. This is because, by Lemma~\ref{lem:parent}, the parent of any node that is lexicographically minimum under the action of $G$  must be lexicographically minimum under the action of $G$. 
On the other hand, if a node (subproblem) obtained by branching 
is  lexicographically minimum under the action of $G$, then it is viable to use $G$ for isomorphism pruning under that node to find all non-$G$-isomorphic optimum solutions. This is because no node that is lexicographically minimum under the action of $G$ will be pruned.
\qed
\end{pf}

%-----------------------------------------------------------------
 \section{Application for classifying orthogonal arrays}\label{sec:OAs}
%----------------------------------------------------------------------
In this section, we use a preliminary parallelization scheme and Method~\ref{alg:implement} with $G=\GLPP$ to classify all non-OD-equivalent $\oant$.
Method~\ref{alg:determine} and Method~\ref{alg:runjobs} were programmed in the {\tt Perl} programming language. Google's {\tt Gemini} AI~\cite{googleGemini3} was used to speed up the programming process.
Method~\ref{alg:implement} was implemented manually.

%We applied a  preliminary parallelization scheme for classifying OA($192,9,2,4$).  
To set up our preliminary parallelization scheme, we set $F^{a_0}\prec   \cdots \prec
F^{a_{p_{\max}}} $ such that $F^{a_i}$ is the ILP obtained from ILP~(\ref{eqn:ILPOAfin}) by replacing $p_{\max}$  with $p_{\max}-i$ and $F^{a_i}$ has  the first variable fixed to $p_{\max}-i$.
%The subproblem $F^{a_{p_{\max}}}$
%while every other variable is less than or equal to $p_{\max}-i$.
 This parallelization works because $\GLPP$ permutes the variables of
 ILP~(\ref{eqn:ILPOAfin}) {\em transitively}. That is, given $x_{i_1}$ and $x_{i_2}$, 
 there is a $g \in \GLPP$ such that $x_{i_1}=x_{g(i_2)}$.
 To see this, it suffices to show that $G^{\rm iso}(k,s)$ permutes the variables of ILP~(\ref{eqn:ILPOAfin}) transitively as $\GLPP\geq G^{\rm iso}(k,s)$. This is obvious from the definition of $G^{\rm iso}(k,s)$.
%Given  variables $x_{(i_1,\ldots i_k)}$ and $x_ {(i'_1,\ldots i'_k)}$, let $g \in G^{\rm iso}(k,s)$ be such that  $g(i_1,\ldots, i_k)=(i'_1,\ldots,i'_k)$. Then, $g x_{(i_1,\ldots,i_k)}= x_{(i'_1,\ldots,i'_k)}$, proving transitivity. 
 As a permutation of objects by a permutation group is a group action on the set of objects, we say $\GLPP$ acts  {\em transitively} on the variables of ILP~(\ref{eqn:ILPOAfin}).
%  further implies that instead of the group  $Stab(F^{a_i},\GLPP)$,
The following theorem justifies using $\GLPP$ with the B\&B with isomorphism pruning algorithm at nodes $F^{a_i}$.
\begin{thm}
 If the group $\GLPP$ is used at  nodes $F^{a_i}$ for $i=0,\ldots, p_{\max}-1$ in the  B\&B with isomorphism pruning algorithm, then the combined set of solutions is 
a set of all lexicographically-$\GLPP$-minimum non-$\GLPP$-isomorphic solutions 
(OA$(N,k,s,t)$s).
\end{thm}
\begin{pf}
    By Theorem~\ref{thm:main}, we can branch on the first variable by fixing it to $p_{\max}-i$
    for $i=0,\ldots,p_{\max}$ and still use $\GLPP$ to find a set of all lexicographically-$\GLPP$-minimum non-$\GLPP$-isomorphic solutions.
     The transitivity of the action of $\GLPP$ implies that no lexicographically $\GLPP$-minimum
 solution can have $x_0=0$ as at least one of the $x_i$s is positive. This rules out the subproblem $F^{a_{p_{\max}}}$.
    By the transitivity of the action of $\GLPP$ on the variables of ILP~(\ref{eqn:ILPOAfin}),
    no lexicographically-$\GLPP$-minimum solution can have $x_0=p_{\max}-i$ and $x_j>p_{\max}-i$
    for some $j>0$.
    \qed
\end{pf}

For OA$(16\lambda,9,2,4)$ and $\lambda\geq 10$,  it is our experience that the preliminary parallelization divides the overall job of classifying  $\oan$ up to OD-equivalence
into $p_{\max}$ jobs $F^{a_i}$ for $i=0,\ldots, p_{\max}-1$ that take nearly equal CPU time and memory. Each of the jobs $F^{a_i}$ is further parallelized by applying Method~\ref{alg:implement} with the group $\GLPP$, where the minimum index breadth first search branching in Method~\ref{alg:determine} starts at the second variable $x_1$.
%-----------------------------------------------------------------------
For OA$(128,9,2,4)$, we ran Method~\ref{alg:implement}  with $m=4$ and $T=.3$ minutes.

The smallest unclassified OA$(N,9,2,4)$ case is  OA$(192,9,2,4)$.
  For classifying all $\GLPP$-non-isomorphic  OA$(192,9,2,4)$, $p_{\max}=3$. So, the preliminary parallelization yields  $F^{a_0}\prec F^{a_1}\prec F^{a_2}$ as subproblems to solve.
 Method~\ref{alg:implement} allows dividing this problem into smaller subproblems that require nearly equal time and memory to solve. 
 %This can be done by examining the time it takes to find 
 %lexicographically $\GLPP$-minimum,   $\GLPP$-non-isomorphic  OA$(176,9,2,4)$ and OA$(160,9,2,4)$ solutions.
 %We can only hope that a parallelization that works well for  OA$(176,9,2,4)$ or OA$(160,9,2,4)$ would also work well for   
%OA$(192,9,2,4)$.
%\item More research is needed to come up with a canonical way of dividing the subproblems at $F^{a_i}$ for $i=0,\ldots,p_{{\max}-1}$ into smaller subproblems while maintaining nearly equal amount of time and memory needs.

\begin{table}[h]
  \centering
  \caption{%The number of OA isomorphism classes, 
  The number of OD-equivalence
classes, real time (minutes), Method~\ref{alg:implement} inputs $(m, T)$.}
  \label{table:example1}
  \begin{tabular}{|cLccc|}
  \hline
    OA($N,k,s,t$) & OA & OD & Time & Parameters \\
    \hline
    OA($128, 9 ,2 ,4$) & 4132 & 680& 63  &  ($4,.3$)\\
      OA($144, 9 ,2 ,4$) & 0 & 0& &  \\
          OA($192, 9 ,2 ,4$) & ?? & ??& &  \\
  %     OA($192, 10 ,2 ,4$) & Value 2 & Value 3& &  \\
  %      OA($192, 11 ,2 ,4$) & Value 2 & Value 3& &  \\
       \hline
  \end{tabular}
\end{table}
We classified $\oantf$ up to OD-equivalence.
If needed, a set of all non-OD-equivalent $\oantf$ can be used  to obtain a set of all non-isomorphic $\oantf$ by following the four steps in~\cite{Bulutoglu2016}.
Each subproblem for the overall parallelization was created and solved using Method~\ref{alg:implement} after implementing the preliminary parallelization.   Hence, for each $\oantf$ case, $m \times p_{\max}$ many threads were used.
Table~\ref{table:example1} reports the number of non-OD-equivalent $\oantf$, the time it took to classify $\oantf$, and the input parameters $m$ and $T$ for Method~\ref{alg:implement}.
Each reported time is the solution time for the longest-running parallel solved subproblem.

Method 5.4 in~\cite{Bulutoglu2016} was used to obtain a set of all 
non-OD-equivalent OA$(192,k+1,2,4)$ from a set of all  OA$(192,k,2,4)$
for $k=9,10$. Table~\ref{table:example2} reports the number of non-OD-equivalent  OA$(192,k,2,4)$ and the time it took to classify  OA$(192,k,2,4)$ for $k=10,11$.
\begin{table}[h]
  \centering
  \caption{%The number of OA isomorphism classes,
  The number of OD-equivalence
classes, real time (minutes).}
  \label{table:example2}
  \begin{tabular}{|cLcc|}
  \hline
    OA($N,k,s,t$) & OA & OD & Time  \\
    \hline
 %   OA($128, 9 ,2 ,4$) & 4132 & 680& 63  &  ($4,.3$)\\
 %     OA($144, 9 ,2 ,4$) & 0 & 0& &  \\
 %         OA($192, 9 ,2 ,4$) & Value 2 & Value 3& &  \\
       OA($192, 10 ,2 ,4$) & ?? & ??&   \\
        OA($192, 11 ,2 ,4$) & ?? & ??&   \\
 %          OA($192, 12 ,2 ,4$) & 0 & 0&   \\     
       \hline
  \end{tabular}
\end{table}

%----------------------------------------------------------------------
\section*{Acknowledgments}
%----------------------------------------------------------------------
During the preparation of this work, the author used Google's {\tt Gemini} AI~\cite{googleGemini3} to write {\tt Perl} code. After using Google's {\tt Gemini}, the author reviewed and checked the  {\tt Perl} code as needed and takes full responsibility for the content of the published article.

The author thanks Ryan Pena and David Doak for their computer support. This research was funded by AFOSR  grant 22RT0446. 
  The views expressed in this article are those of the author and do not reflect the official policy or position of the United States Air Force, Department of Defense, or the U.S.~Government.

%-----------------------------------------------------------------------
%-----------------------------------------------------------------------
\bibliographystyle{elsarticle-harv}
\bibliography{BibliographyDO}
%--------------------------------------------------------

%-----------------------------------------------------------------------
\end{document}